\documentclass[letterpaper, 10 pt, conference]{ieeeconf} 
\IEEEoverridecommandlockouts

\usepackage{amsmath}
\usepackage{color}
\usepackage{algorithm}
\usepackage{tikz}
\usepackage{graphicx}
\usepackage{amsfonts}
\usepackage[noadjust]{cite}
\usepackage{balance} 
\usepackage[noend]{algpseudocode}
\usepackage{tikz}
\usetikzlibrary{arrows,automata,positioning}
\usepackage[noadjust]{cite}

\newtheorem{definition}{Definition}
\newtheorem{example}{Example}
\newtheorem{problem}{Problem}

\newtheorem{theorem}{Theorem}
\newtheorem{corollary}{Corollary}
\newtheorem{proposition}{Proposition}

\title{\Large \bf Formal Methods for Adaptive Control of Dynamical Systems}

\author{Sadra Sadraddini and Calin Belta 
\thanks{The authors are with the Department of Mechanical Engineering, Boston University, Boston, MA 02215  \{sadra,cbelta\}@bu.edu.
This work was partially supported by the NSF under grants CPS- 1446151 and CMMI-1400167.}
}
\begin{document}
\maketitle

\thispagestyle{empty}
\pagestyle{empty}

\begin{abstract}

We develop a method to control discrete-time systems with constant but initially unknown parameters from linear temporal logic (LTL) specifications. We introduce the notions of (non-deterministic) parametric and adaptive transition systems and show how to use tools from formal methods to compute adaptive control strategies for finite systems. For infinite systems, we first compute abstractions in the form of parametric finite quotient transition systems and then apply the techniques for finite systems. Unlike traditional adaptive control methods, our approach is correct by design, does not require a reference model, and can deal with a much wider range of systems and specifications. Illustrative case studies are included.

\end{abstract}

\section{Introduction}
%
%

Adaptive control, or self-learning control, is a set of techniques to automatically adjust controllers for uncertain systems. In the traditional problem of adaptive control, a parameterized system is considered where the parameters are assumed to be constant, but their values are initially unknown to the controller. The goal is to achieve some desired performance while the parameters are (possibly indirectly) estimated online. The solution to this problem can be extended to scenarios where parameters infrequently change or vary slowly. Numerous adaptive control methods have been developed since 1950s \cite{aastrom2013adaptive,krsticnonlinear,slotine1991applied,ioannou2012robust}. The main theoretical guarantee sought in all conventional adaptive control techniques is stability - whether it is specified in terms of tracking a set-point, trajectory, or a reference model.

One particular limitation of current adaptive control methods is handling systems that involve discontinuities.  Most adaptive control techniques rely on the continuity of the model and its parameterization. In many realistic models, state, control or parameters take values from both continuous and discrete domains. Within methods that do not entirely depend on the continuity of the model, a promising direction is using multiple models/controllers \cite{morse1996supervisory,narendra2000adaptive,anderson2001multiple,hespanha2003overcoming}, where the objective is to achieve stability via designing a switching law (supervisory control) to coordinate the controllers. Model reference adaptive control (MRAC) of specific forms of scalar input piecewise affine systems were studied in \cite{di2013hybrid,di2016extended}. However, it is still not clear how to deal with general discrete or hybrid systems.  


Another remaining open problem in adaptive control is dealing with specifications richer than stability. In many engineering applications, we are interested in complex requirements composed of safety (something bad never happens), liveness (something good eventually happens), sequentiality of tasks, and reactiveness. Temporal logics \cite{baier2008principles} provide a natural framework for specifying such requirements. The main challenge in designing adaptive control techniques from formal specifications is handling hard constraints on the evolution of the system. Even for the simpler problem of constraints defined as a safe set in the state-space, designing adaptive control strategies is challenging. Existing works on this problem \cite{guay2012adaptive,aswani2013provably,tanaskovic2014adaptive,di2016indirect,he2016adaptive} apply robust control techniques to ensure infinite-time constraint satisfaction for all admissible parameters. This approach may be severely conservative since if a robust control strategy does not exist for all admissible parameters, it does not necessarily indicate that constraints can not be satisfied after some measurements are taken from the system and a more accurate model is available. Even though  \cite{aswani2013provably,tanaskovic2014adaptive,di2015indirect} update the model and synthesize controls in a receding horizon manner, they decouple constraint satisfaction and learning. However, there exists a deep coupling: when synthesizing controls, not only constraints must be taken into account, but also the evolution of the system should also lead to subsequent measurements that are more informative about the uncertainties in the model. In other words, control decisions have a indirect influence on the way the model is updated. 
We use tools from formal methods \cite{clarke1999model,baier2008principles} to develop a framework for correct-by-design adaptive control that can deal with complex systems and specifications. Formal methods have been increasingly used in control theory in recent years \cite{tabuada2009verification,belta2017book}. We consider discrete-time systems with constant but initially unknown parameters. We describe system specifications using linear temporal logic (LTL) \cite{baier2008principles}. As in any other adaptive control technique, we require an online parameter estimator. Our parameter estimator maps the history of the evolution of the system to the set of ``all possible" parameters, which contains the actual parameters. We embed the parameterized system in a (non-deterministic) parametric transition system (PTS), from which we construct a (non-deterministic) adaptive transition system (ATS) that contains all the possible combinations of transitions with the unfoldings of the parameter estimator. 
The main results and contributions of this paper are as follows:
\begin{itemize}
\item For finite systems, the LTL adaptive control problem reduces to a Rabin game \cite{thomas2002automata} on the product of the finite ATS and the Rabin automaton corresponding to the LTL specification. 
 The method is correct by design and it is complete, i.e. it finds a solution if one exists;
\item For infinite systems, we construct finite quotient ATSs by partitioning the state and the parameter space and quantizing the control space. Once an adaptive control strategy is found for the quotient, it is guaranteed that it will also ensure the satisfaction of the LTL formula for the original infinite system. The method may be conservative. 
\end{itemize}

This paper is related to recent works that seek a formal approach to combining learning and control. The authors in \cite{quindlen2016region,kozarev2016case} provided statistical certificates for MRAC subject to safety constraints. The idea was based on implementing MRAC from a set of different initial conditions and parameters and observing if the trajectories were safe.  However, the design of MRAC itself did not take into account the constraints. Moreover, given a temporal logic specification and a system model with parametric uncertainty, it is not clear how a reference model should be chosen for implementing MRAC. If a reference model is able to satisfy the specification, the matching condition may not hold, i.e. there may not exist a controller for the original system to behave as the reference model. Therefore, classic MRAC may not be suitable for the purpose of this paper as it requires a careful search of reference models subject to matching conditions.

Reinforcement learning (RL) methods are conceptually similar to adaptive control, but are used in a probabilistic framework and require a reward mechanism to generate control policies. The authors in \cite{sadigh2014learning} studied RL from LTL specifications, where large rewards were dedicated to the pairs in the Rabin automaton to incentivize the system to visit them regularly or avoid them. In \cite{aksaray2016q}, Q-learning was applied to control MDPs from signal temporal logic (STL) specifications, where the reward was the STL robustness score - a measure of distance to satisfaction. Other closely related works include \cite{fu2014adaptive,leahy2016integrate}, where the problem of LTL control was modeled as a game between a player (controller) and an adversary (environment). The controller inferred the ``grammar" of actions taken by the environment.  However, this approach also decoupled adaptation (learning) and control. If the LTL formula was violated during the grammar learning, the control software stopped. While these methods (including RL) have the advantage that they require less prior knowledge about the system, they are not suitable for performance-critical systems with constraints that should never be violated, even during the learning process.



This paper is organized as follows. First, we provide the necessary background on LTL, transition systems and LTL control in Sec. \ref{sec:back}. The adaptive control problem is formulated in Sec. \ref{sec:problem}. We define PTSs in Sec. \ref{sec:parametric}. Technical details for the solutions for finite and infinite systems are explained in Sec. \ref{sec:finite} and \ref{sec:infinite}, respectively. Finally, two case studies are presented in Sec. \ref{sec:case}.

\section{Background}
\label{sec:back}
\subsection{Notation}
The set of real and Boolean values are denoted by $\mathbb{R}$ and $\mathbb{B}$ respectively. The empty set is denoted by $\emptyset$. Given a set ${S}$, we use $|S|$,  $2^S$, $2^S_{-\emptyset}$ to denote its cardinality, power set, and power set excluding the empty set, respectively. 
An alphabet $\mathcal{A}$ is a finite set of symbols $\mathcal{A}=\{a_1,a_2,\cdots,a_A\}$. A finite (infinite) word is a finite-length (infinite-length) string of symbols in $\mathcal{A}$. 
For example, $w_1=a_1a_2a_1$ is a finite word, and $w_2=a_1a_2\overline{a_1}$ and $w_3=a_1\overline{a_2a_1}$ are infinite words over $\mathcal{A}=\{a_1,a_2\}$, where over-line stands for infinitely many repetitions. 
We use $\mathcal{A}^*$ and $\mathcal{A}^\omega$ to denote the set of all finite and infinite words that can be generated from $\mathcal{A}$, respectively.

\subsection{Linear Temporal Logic}
The formal definition of LTL syntax and semantics is not provided here as it can be found in the literature \cite{baier2008principles}. Here we provide an informal introduction and the necessary notation. LTL consists of a finite set of atomic propositions $\Pi$, temporal operators ${\bf G}$ (globally/always), ${\bf F}$ (future/eventually), ${\bf U}$ (Until), and Boolean connectives $\wedge$ (conjucntion), $\vee$ (disjunction), and $\neg$ (negation). LTL semantics are interpreted over infinite words over $2^\Pi$.
The set of all infinite words that satisfy an LTL formula $\varphi$ is denoted by $L(\varphi)$, $L(\varphi) \subset (2^\Pi)^\omega$, and is referred to as the \emph{language} of $\varphi$.

\begin{definition}
A Deterministic Rabin Automaton (DRA) is defined as the tuple $\mathcal{R}=(S,s^0,\mathcal{A},\alpha,\Omega)$, where:
\begin{itemize}
\item $S$ is a set of states;
\item $s^0$ is the initial state;
\item $\mathcal{A}$ is a finite set of inputs (alphabet);
\item $\alpha$ is a transition function $\alpha:S \times \mathcal{A} \rightarrow S$;
\item $\Omega=\left \{(F_1,I_1),\cdots,(F_r,I_r) \right\}$ is a finite set of pairs of sets of states, where $F_i,I_i \subset S ,i=1,\cdots,r$.    
\end{itemize}
\end{definition}

An infinite word $w \in \mathcal{A}^\omega$ determines a sequence of inputs for $\mathcal{R}$ that results in the \emph{run} $\zeta(w)=s_0s_1\cdots$, where $s_{k+1}=\alpha(s_k,a_k)$, $s_0=s^0$, and $a_k$ is the $k$'th input appearing in $w$. We define $Inf(\zeta)=\left\{ s | s \text{ appears infinitely often in } \zeta \right\}$. A run $\zeta$ is \emph{accepted} by $\mathcal{R}$ if there exists $i \in \{1,\cdots,m\}$ such that $Inf(\zeta) \cap F_i = \emptyset$ and $Inf(\zeta) \cap I_i \neq \emptyset$. In other words, $F_i$ is visited finitely many times and $I_i$ is visited infinitely often for some $i$. The language of $\mathcal{R}$, denoted by $L(\mathcal{R})$, $L(\mathcal{R}) \subset \mathcal{A}^\omega$ , is defined as the set of all elements in $\mathcal{A}^\omega$ that produce accepting runs.   

It is known that given an LTL formula $\varphi$ over $\Pi$, one can construct a DRA $\mathcal{R}_\varphi$ with input set $\mathcal{A}=2^\Pi$ such that $L(\mathcal{R}_\varphi)=L(\varphi)$ \cite{thomas2002automata}. Therefore, verifying whether an infinite word satisfies an LTL formula becomes equivalent to checking the Rabin acceptance condition. There exists well-established algorithms and software for this procedure \cite{klein2006experiments}.

\begin{example}
Consider $\varphi={\bf G F} \pi_1 \wedge {\bf F} \pi_2$, which is an LTL formula over $\Pi=\{\pi_1,\pi_2\}$, stating that ``$\pi_1$ holds infinitely often, and $\pi_2$ eventually holds". The DRA $\mathcal{R}_\varphi$ corresponding to this formula is illustrated in Figure \ref{fig:rabin}. For example, we have $\{\pi_2\}\overline{\{\pi_1,\pi_2\}} \models \varphi$ ($\varphi$ is satisfied), but $\overline{\{\pi_1\}} \not \models \varphi$ ($\varphi$ is violated since $\pi_2$ never appears), and $\{\pi_1\}\overline{\emptyset \{\pi_2\} } \not \models \varphi$ (because $\pi_1$ does not hold infinitely often). 
\end{example}

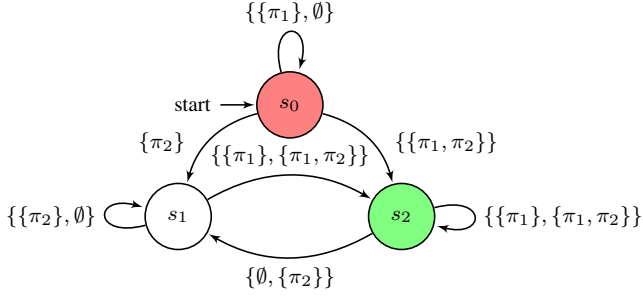
\begin{figure}
\centering
\vspace{0.2in}
\begin{tikzpicture}[>=latex',shorten >=0.5pt,node distance=2.1cm,on grid,auto,semithick]
\tikzset{font={\fontsize{8pt}{8}\selectfont}}
\node[state,initial,fill=red!50] (s0) {$s_0$};
  \node[state,fill=green!50] (s2) [below right=of s0] {$s_2$};
  \node[state] (s1) [below left=of s0] {$s_1$};
  \path[->] (s0) edge [bend right,pos=0.8,above left] node {$\{\pi_2\}$} (s1);
  \path[->] (s0) edge [loop above] node {$\{ \{\pi_1\},\emptyset\}$} (s0);
  \path[->] (s0) edge [bend left,pos=0.8] node {$\{ \{\pi_1,\pi_2\} \}$} (s2);
  \path[->] (s1) edge [bend left,pos=0.5] node {$\{\{\pi_1\},\{\pi_1,\pi_2\} \}$} (s2);
  \path[->] (s1) edge [loop left] node {$\{ \{\pi_2 \},\emptyset \}$} (s1);
  \path[->] (s2) edge [loop right] node {$\{ \{\pi_1 \},\{\pi_1,\pi_2\} \}$} (s2);
  \path[->] (s2) edge [bend left,pos=0.5] node {$\{ \emptyset,\{\pi_2 \} \}$} (s1);
\end{tikzpicture}
\caption{Example 1: DRA corresponding to $\varphi={\bf G F} \pi_1 \wedge {\bf F} \pi_2$, where $F_1=\{s_0\}$ (red), $I_1=\{s_2\}$ (green). Runs that visit the green state infinitely many times and visit the red state finitely many times satisfy $\varphi$.}
\label{fig:rabin}
\end{figure}

\subsection{Transition Systems}
\label{sec:transit}
\begin{definition}
A transition system is defined as the tuple $\mathcal{T}=\left(X, U, \beta, \Pi, O \right)$, where:

\begin{itemize}
\item $X$ is a (possibly infinite) set of states;
\item $U$ is a (possibly infinite) set of control inputs;
\item $\beta$ is a transition function $\beta:X \times U \rightarrow 2^X$;
\item $\Pi=\{\pi_1,\pi_2,\cdots,\pi_{m}\}$ is a finite set of atomic propositions;
\item $O: X \rightarrow 2^\Pi$ is an observation map.
\end{itemize}
\end{definition}
We assume that $\mathcal{T}$ is \emph{non-blocking} in the sense that $|\beta(x,u)| \neq 0$ for all $x \in X, u \in U$. 
\footnote{
If $\mathcal{T}$ is blocking, we can make it non-blocking by adding an additional state $x^{sink}$ such that for all $x \in X, u \in U, |\beta(x,u)| = 0$, we have $x^{sink}=\beta(x,u)$. Also, we add transitions $x^{sink}=\beta(x^{sink},u), \forall u \in U$. In order to prevent blocking, we find a control strategy such that $x^{sink}$ is not reachable.  
}
A transition system $\mathcal{T}$ is \emph{deterministic} if $|\delta(x,u)| = 1, \forall x \in X, \forall u \in U$, and   is \emph{finite} if $X$ and $U$ are finite sets. 
 A trajectory of $\mathcal{T}$ is an infinite sequence of visited states $x_0x_1x_2\cdots$. The infinite {word} produced by such a trajectory is $O(x_0)O(x_1)O(x_2)\cdots$. Note that the alphabet here is $2^\Pi$. 
The set of all infinite words that can be generated by $\mathcal{T}$ is a subset of ${(2^\Pi)}^\omega$.
 
\begin{definition}
A {control strategy} $\Lambda$ is a function $\Lambda:X^* \times U^* \rightarrow U$ that maps the history of visited states and applied controls to an admissible control input, where $u_k=\Lambda(x_0\cdots,x_k,u_0\cdots,u_{k-1}), \forall k\in \mathbb{N}$.   
\end{definition}

\begin{definition}
Given a transition system $\mathcal{T}=\left(X, U, \beta, \Pi, O \right)$, a control strategy $\Lambda$ and a set of initial states $X_0 \in X$, we define:
\begin{equation*}
\begin{array}{ll}
L(\mathcal{T},\Lambda,X_0):=\Big\{ & O(x_0)O(x_1)\cdots \in {(2^\Pi)}^\omega \Big | 
\\
& x_0 \in X_0, x_{k+1} \in \beta(x_k,u_k), k \in \mathbb{N} \Big \},
\end{array}
\end{equation*}
where $u_k=\Lambda(x_0\cdots,x_k,u_0\cdots,u_{k-1})$. 
\end{definition}

%
%
%
    
\subsection{Quotient Transition System}
\label{sec:quotient}

Consider a transition system $\mathcal{T}=\left(X, U, \beta, \Pi, O \right)$. 
A (finite) set $Q \subset 2^X$ is a (finite) partition for $X$ if 1) $\emptyset \not \in Q,$ 2) $\bigcup_{q\in Q}q=X$, and 3) $q \cap q'= \emptyset, \forall q,q' \in Q, q\neq q'$. 
A partition $Q$ is \emph{observation preserving} if for all $q\in Q$, we have $O(x)=O(x'), \forall x,x' \in q$.

\begin{definition}
\label{def:quotient}
Given a transition system $\mathcal{T}=\left(X, U, \beta, \Pi, O \right)$ and an observation preserving partition $Q$ for $X$, the \emph{quotient transition system} is defined as the tuple $\mathcal{T}_Q=\left(Q, U, \beta_Q, \Pi, O_Q \right)$ such that:
\begin{itemize}
\item  for all $q\in Q$, we have $q' \in \beta_Q(q,u)$ if and only if $\exists x \in q$, $\exists x' \in q'$ such that $x' \in \beta(x,u)$;
\item  for all $q\in Q$, we have $O_Q(q)=O(x)$ for any $x \in q$.
\end{itemize} 
\end{definition}
Given a control strategy for the quotient $\Lambda_Q:Q^* \times U^* \rightarrow U$, and a set of initial conditions $Q_0$, we construct $\Lambda^{(Q)}:X^* \rightarrow U$ such that $\Lambda^{(Q)}(x_0\cdots x_k)=\Lambda_{Q}(q_0\cdots q_k)$, $x_i \in q_i$, $0\le i \le k$, $k \in \mathbb{N}$, and $X_0^{(Q)}=\{x_0| x_0 \in q_0, q_0 \in Q_0\}$. It is easy to show that $L(\mathcal{T},\Lambda^{(Q)},X_0^{(Q)}) \subseteq L(\mathcal{T}_Q,\Lambda_Q,Q_0)$, which stems from the fact that $\mathcal{T}_Q$ simulates $\mathcal{T}$. We refer to $L(\mathcal{T}_Q,\Lambda_Q,Q_0) \setminus L(\mathcal{T},\Lambda^{(Q)},X_0^{(Q)})$ as the set of spurious infinite words (SIW). 
In order to have $L(\mathcal{T},\Lambda^{(Q)},X_0^{(Q)}) = L(\mathcal{T}_Q,\Lambda_Q,Q_0)$ (empty SIW), a sufficient condition is that $\mathcal{T}_Q$ and $\mathcal{T}$ are \emph{bisimilar} \cite{belta2017book}.  
For infinite $X$, there is no general guarantee that a finite $Q$ exists such that $\mathcal{T}_Q$ is bisimilar to $\mathcal{T}$. In order to ``shrink" SIW, $Q$ is refined. At the most extreme case, SIW remains nonempty unless $Q=X$. Further details on simulation and bisimulation relations are not required for this paper and the interested reader is referred to the related works in the literature, such as \cite{fernandez1991fly,tabuada2009verification,belta2017book}.

%
%
%
%
%
%
%


\subsection{LTL Control}
\label{sec:ltlcon}
Given a finite transition system $\mathcal{T}=\left(X, U, \beta, \Pi, O \right)$ and an LTL formula $\varphi$ over $\Pi$, we are interested in finding a control strategy $\Lambda$ and the largest set of initial conditions $X_0^{\max}$ such that $L(\mathcal{T},\Lambda,X_0^{\max}) \subseteq L(\varphi)$. In other words, we require $\varphi$ to be satisfied for all trajectories that are allowed by the non-determinism in $\mathcal{T}$.

\begin{definition} 
Given a transition system $\mathcal{T}=\left(X, U, \beta, \Pi, O \right)$ and a DRA $\mathcal{R}_\varphi=(S,s^0,\mathcal{A},\alpha,\Omega)$ corresponding to LTL formula $\varphi$, the product automaton $\mathcal{T}_\varphi^P=\mathcal{T} \otimes \mathcal{R}_\varphi$ is  defined as the tuple $\left(X^P,X^{P,0},U, \beta^P, \Omega^P \right)$, where:

\begin{itemize}
\item $X^P=X \times S$ is the set of product states;
\item $X^{P,0}=\{(x,s^0) | x\in X\}$ is the set of initial product states;
\item $U$ is the set of control inputs;
\item $\beta^P:X^P \times U \rightarrow 2^{X^P}$ is the product transition function, where $x^{P'}\in \delta(x^P,u)$, $x^P=(x,s),x^{P'}=(x',s')$, if and only if $x' \in \beta(x,u)$ and $s'=\alpha(s,O(x))$. 
\item $\Omega^P=\left \{(F^P_1,I^P_1),\cdots,(F^P_r,I^P_r) \right\}$ is a finite set of pairs of sets of states, where $F^P_i=\{(x,s) | x\in X, s \in F_i\},I^P_i=\{(x,s) | x\in X, s \in I_i\} ,i=1,\cdots,r$.    
\end{itemize}
\end{definition}

The product automaton $\mathcal{T}_\varphi^P$ is a (non-deterministic) automaton (with control inputs) capturing both the transitions in $\mathcal{T}$ and the acceptance condition of $\varphi$. The solution to the problem of finding a control strategy to satisfy $\varphi$ is accomplished by solving the Rabin game on the product automaton. The details are not presented here but can be found in \cite{chatterjee2012survey}. It can be shown that the control strategy is memoryless on the product automaton in the form $\Lambda: X \times S \rightarrow U$. In other words, the history of the system is incorporated into the state of the Rabin automaton.
The largest set of admissible initial conditions $X_0^{\max}$ corresponds to the winning region of the Rabin game. 

If the transition system $\mathcal{T}$ is infinite, a finite quotient is constructed. If $U$ is infinite, it can be quantized to obtain a finite set
\footnote{An alternative (better) approach was proposed in \cite{Yordanov2012} for piecewise affine systems, where the authors computed a finite set of sets of control inputs that enabled transitions with minimal non-determinism in the quotient system.}.
It is straightforward to show that if a control strategy satisfying $\varphi$ exists for the finite quotient, it also satisfies $\varphi$ if implemented on the original system. However, unless the quotient and the original transition system are bisimilar, the non-existence of a control strategy for the quotient does not indicate that one does not exist for the original system.   
Hence the approach of using finite quotients may be conservative \cite{tabuada2009verification,belta2017book}.

\section{Problem Formulation and approach}
\label{sec:problem}
We are interested in discrete-time systems of the following form: 
\begin{equation}
\label{eq:system}
\begin{array}{rl}
x^+= &F(x,u,\theta,d), \\
y_i= & \mu_i(x), i=1,\cdots,m,
\end{array}
\end{equation}
where $x\in X$ is the state, $u \in U$ is the control input, $\theta \in \Theta$ represents the parameters of the system, $d \in D$ is the disturbance (adversarial input), $F:X \times U \times \Theta \times D \rightarrow X$ is the system evolution function, and $y_i, i=1,\cdots,m$, are Boolean system outputs, where $\mu_i: X \rightarrow \mathbb{B}$. 
We define the set of atomic propositions $\Pi=\{\pi_1,\cdots,\pi_m\}$ such that $x \models \pi_i \Leftrightarrow \mu_i(x)=\text{True}, i=1,\cdots, m$.
The sets $X,U,\Theta,D$ are the admissible sets for states, controls, parameters and disturbances respectively. All sets may be finite or infinite. System \eqref{eq:system} is finite if $X,U,\Theta,D$ are all finite.  


\begin{example}
A prominent class of systems encountered in adaptive control are parameterized linear systems, where $F(x,u,\theta,d)=A(\theta)x+B(\theta)u+d$. We have $X \subset \mathbb{R}^{n_x}$, $U \subset \mathbb{R}^{n_u}$, $\Theta \subset \mathbb{R}^{n_\theta}$, $D \subset \mathbb{R}^{n_d}$. $A,B$ are matrices with appropriate dimensions that depend on $\theta$. It is also common to assume that the outputs are Boolean evaluations of linear predicates $\mu_i=(r_i^T x \le \rho_i)$, where $r_i \in \mathbb{R}^n,$ and $\rho_i \in \mathbb{R}$. Thus, each proposition $\pi_i$ defines a closed half space in $\mathbb{R}^{n_x}$.  
\end{example} 

  
As mentioned in the introduction, we distinguish between the uncertainty in parameters and disturbances. Disturbances usually have unknown (fast) variations in time. In this paper, we assume that $\theta$ is a constant but its value $\theta^*$ is initially unknown. If we treat the uncertainties in  parameters and disturbances in the same way, we are required to design control strategies that are robust versus  all values in both $\Theta$ and $D$. This approach is severely conservative and often fails to find a solution.  
The key idea of adaptive control is to take advantage of the fact that $\theta^*$ can be (approximately) inferred from the history of the evolution of the system. Therefore, adaptive control is often significantly more powerful than pure robust control and it is also more difficult to design and analyze. 
In engineering applications, parameters are related to the physical attributes of the plant whereas disturbances are related to effects of stochastic nature such as imperfect actuators/sensors and perturbations in the environment. 

\begin{problem}
\label{problem}
Given system \eqref{eq:system} and an LTL formula $\varphi$ over $\Pi$, find a control strategy $\Lambda: X^* \times U^* \rightarrow U$ and a set of initial states $X_0 \subseteq X$ such that all the trajectories of the closed loop system starting from $X_0$ satisfy $\varphi$.  
\end{problem}
\vspace{0.2in}
{
Our aim is to convert Problem \ref{problem} to an LTL control problem described in  Sec.\ref{sec:ltlcon} and use the standard tools for Rabin games. To this end, we need to incorporate adaptation into control synthesis.  
The central tool to any adaptive control technique is parameter estimation. Note that an adaptive control strategy has the form $\Lambda: X^* \times U^* \rightarrow U$, since parameters are estimated using the history of the evolution of the system. 
We take the following approach to convert Problem \ref{problem} into an LTL control problem. We embed system \eqref{eq:system} in a parametric transition system (PTS), which is defined in Sec. \ref{sec:parametric}. We construct a finite adaptive transition system (ATS) from a finite PTS. An ATS is an ordinary transition system as in Sec. \ref{sec:transit}, but parameters are also incorporated into its states and transitions in appropriate way, which is explained in Sec. \ref{sec:finite}. We deal with an infinite PTS by constructing a finite quotient PTS in Sec. \ref{sec:infinite}. 
}

\section{Parametric Transition System}
\label{sec:parametric}

\begin{definition}
A parametric transition system (PTS) is defined as the tuple $\mathcal{T}^\Theta=\left(X, U, \Theta, \gamma, \Pi, O \right)$, where:
\begin{itemize}
\item $X$ is a (possibly infinite) set of states;
\item $U$ is a (possibly infinite) set of control inputs;
\item $\Theta$ is a (possibly infinite) set of parameters;
\item $\gamma$ is a transition function $\gamma:X \times U \times \Theta \rightarrow 2^X$. 
\item $\Pi=\{\pi_1,\pi_2,\cdots,\pi_{m}\}$ is a finite set of atomic propositions;
\item $O: X \rightarrow 2^\Pi$ is an observation map.
\end{itemize}
\end{definition}
The only difference between a PTS and a transition system is that its transitions depend on parameters. Note that if $|\Theta|=1$, a PTS becomes a transition system. 

Now we explain how to represent \eqref{eq:system} in the form of a PTS. The sets $X,U,\Theta$ are inherited from \eqref{eq:system} (which is why we have used the same notation). The transition function $\gamma$ is constructed such that 
\begin{equation}
\gamma(x,u,\theta)=\left \{ F(x,u,\theta,d) \Big| d \in D \right\}.
\label{eq:embed_parameteric}
\end{equation}  
The observation map $O:X \rightarrow 2^\Pi$ is given by:
\begin{equation}
O(x)=\left \{\pi_i \Big|\mu_i(x)=\text{True}, i=1\cdots,m \right\}.
\end{equation} 
Therefore, $\mathcal{T}^\Theta=\left(X, U, \Theta, \gamma, \Pi, O \right)$ captures everything in system \eqref{eq:system}. We refer to $\mathcal{T}^\Theta$ as the \emph{embedding} of \eqref{eq:system}.   
One can interpret a PTS as a (possibly infinite) family of transition systems. The actual transitions are governed by a single parameter $\theta^*$, which is initially unknown to the controller. Therefore, the controller has to find out which transition system is the ground truth. 

\section{Control Synthesis for Finite Systems}
\label{sec:finite}

In this section, we assume the PTS embedding system \eqref{eq:system} is finite.

\subsection{Parameter Estimation}
\begin{definition}
A \emph{parameter estimator} $\Gamma$ is a function 
\begin{equation}
\label{eq:estimator}
\Gamma: X^* \times U^* \rightarrow 2^\Theta_{-\emptyset}
\end{equation} that maps the history of visited states and applied controls to a subset of parameters. We have $\vartheta_k=\Gamma(x_0 \cdots x_k; u_0\cdots u_{k-1})$, where:  
\begin{equation}
\label{eq:estimator}
\vartheta_k = \left \{ \theta \in \Theta \Big | x_{i+1} \in \gamma(x_i,u_i,\theta), 0\le i \le k-1 \right\}.
\end{equation}    
\end{definition}
One can see that the parameter estimator \eqref{eq:estimator} is ``{sound}" in the sense that $\theta^* \in \vartheta_k, \forall k\in \mathbb{N}$. We have $\vartheta_0=\Gamma(x_0)=\Theta$, by definition. Note that our definition of parameter estimator is different from the traditional ones, which are often in the form $X^* \times U^* \rightarrow \Theta$, as they return only an estimate $\hat{\theta}$ rather than the set of all possible parameters. For our formal setup, it is vitally important that the controller take into account all possible ground truth parameters at all times. Otherwise, guaranteeing the specification is impossible. The following proposition enables us to make \eqref{eq:estimator} recursive.


\begin{proposition}
The following recursive relation holds:
\begin{equation}
\vartheta_{k+1}=\left \{ \theta \in \vartheta_k \Big | x_{k+1} \in \gamma(x_k,u_k,\theta) \right\}.
\end{equation}
\end{proposition}
\begin{proof}
Substitute $\vartheta_k$ from \eqref{eq:estimator}:
\begin{equation*}
\begin{array}{rl}
& \left \{ \theta \in \vartheta_k \Big | x_{k+1} \in \gamma(x_k,u_k,\theta) \right\} \\
= & \left \{ \theta \in \Theta \Big | \theta \in \vartheta_k, x_{k+1} \in \gamma(x_k,u_k,\theta) \right\} \\
= & \left \{ \theta \in \Theta \Big | x_{i+1} \in \gamma(x_i,u_i,\theta), 0\le i \le k \right\} = \vartheta_{k+1}.
\end{array}
\end{equation*}
\end{proof}
\begin{corollary}
The set of estimated parameters never grows: $\vartheta_{k+1} \subseteq \vartheta_k, \forall k\in \mathbb{N} $.
\end{corollary}
Therefore, we obtain a recursive parameter estimator $\Gamma_{rec}:2^\Theta_{-\emptyset} \times X \times U \times X  \rightarrow 2^\Theta_{-\emptyset}$ as $\vartheta_{k+1}=\Gamma_{rec}(\vartheta_k,x_k,u_k,x_{k+1})$. Note that $\Gamma_{rec}$ is deterministic.


\begin{figure*}
\centering
\begin{tikzpicture}[>=latex',shorten >=1pt,node distance=1.5cm,on grid,auto]
  \node[state] (s0) {$x_1$};
  \node[state] (s1) [below right=of s0] {$x_2$};
  \node[state] (s2) [below left=of s0] {$x_3$};
  \path[->] (s0) edge node {$u_1$} (s1);
  \path[->] (s0) edge [bend left] node {$u_2$} (s2);
  \path[->] (s1) edge [loop right] node {$u_1,u_2$} (s1);
  \path[->] (s2) edge [loop left] node {$u_1$} (s2);
  \path[->] (s2) edge [bend left] node {$u_2$} (s0);
\end{tikzpicture}
~~~~~~~
\begin{tikzpicture}[>=latex',shorten >=1pt,node distance=1.5cm,on grid,auto]
  \node[state] (s0) {$x_1$};
  \node[state] (s1) [below right=of s0] {$x_2$};
  \node[state] (s2) [below left=of s0] {$x_3$};
  \path[->] (s0) edge node {$u_1$} (s1);
  \path[->] (s0) edge [bend right] node {$u_2$} (s2);
  \path[->] (s1) edge [] node {$u_1$} (s2);
  \path[->] (s1) edge [loop right] node {$u_2$} (s1);
  \path[->] (s2) edge [loop left] node {$u_1,u_2$} (s2);
\end{tikzpicture}
\\
{$\mathcal{T}^{\theta_1}$ \hspace{2.6in} $\mathcal{T}^{\theta_2}$}
\\
\vspace{0.2in}
\begin{tikzpicture}[>=latex',shorten >=1pt,node distance=3.2cm,on grid,auto]
  \node[state] (s1) {$x_1,\{\theta_1,\theta_2\}$};
  \node[state] (s2) [right=of s1] {$x_2,\{\theta_1,\theta_2\}$};
  \node[state] (s3) [left=of s1] {$x_3,\{\theta_1,\theta_2\}$};
  \node[state] (s11) [above=of s1] {$x_1,\{\theta_1\}$};
  \node[state] (s21) [above right=of s1] {$x_2,\{\theta_1\}$};
  \node[state] (s31) [above left=of s1] {$x_3,\{\theta_1\}$};
  \node[state] (s12) [below=of s1] {$x_1,\{\theta_2\}$};
  \node[state] (s22) [below right=of s1] {$x_2,\{\theta_2\}$};
  \node[state] (s32) [below left=of s1] {$x_3,\{\theta_2\}$};
  \path[->] (s1) edge node {$u_1$} (s2);
  \path[->] (s1) edge node {$u_2$} (s3);
  \path[->] (s2) edge [loop right] node {$u_2$} (s2);
  \path[->] (s2) edge node {$u_1$} (s21);
  \path[->] (s2) edge [sloped] node {$u_1$} (s32);
  \path[->] (s3) edge [loop left] node {$u_1$} (s3);
  \path[->] (s3) edge node {$u_2$} (s32);
  \path[->] (s3) edge [bend right, sloped] node {$u_2$} (s11);
  \path[->] (s11) edge [] node {$u_1$} (s21);
        \path[->] (s11) edge [bend left] node {$u_2$} (s31);
  \path[->] (s21) edge [loop right] node {$u_1,u_2$} (s21);
  \path[->] (s31) edge [loop left] node {$u_1$} (s31);
  \path[->] (s31) edge [sloped,bend left] node {$u_2$} (s11);
  \path[->] (s12) edge [] node {$u_1$} (s22);
        \path[->] (s12) edge [] node {$u_2$} (s32);
  	\path[<-] (s32) edge [] node {$u_1$} (s22);
  \path[->] (s22) edge [loop right] node {$u_2$} (s22);
  \path[->] (s32) edge [loop left] node {$u_1,u_2$} (s32);

\end{tikzpicture}
\\
{$\mathcal{T}^{adp}$}
\caption{Example 3: [Top] A PTS with two possible parameters $\theta_1,\theta_2$, and the corresponding transition systems [Bottom] The corresponding ATS }
\label{fig:ATS}
\end{figure*}
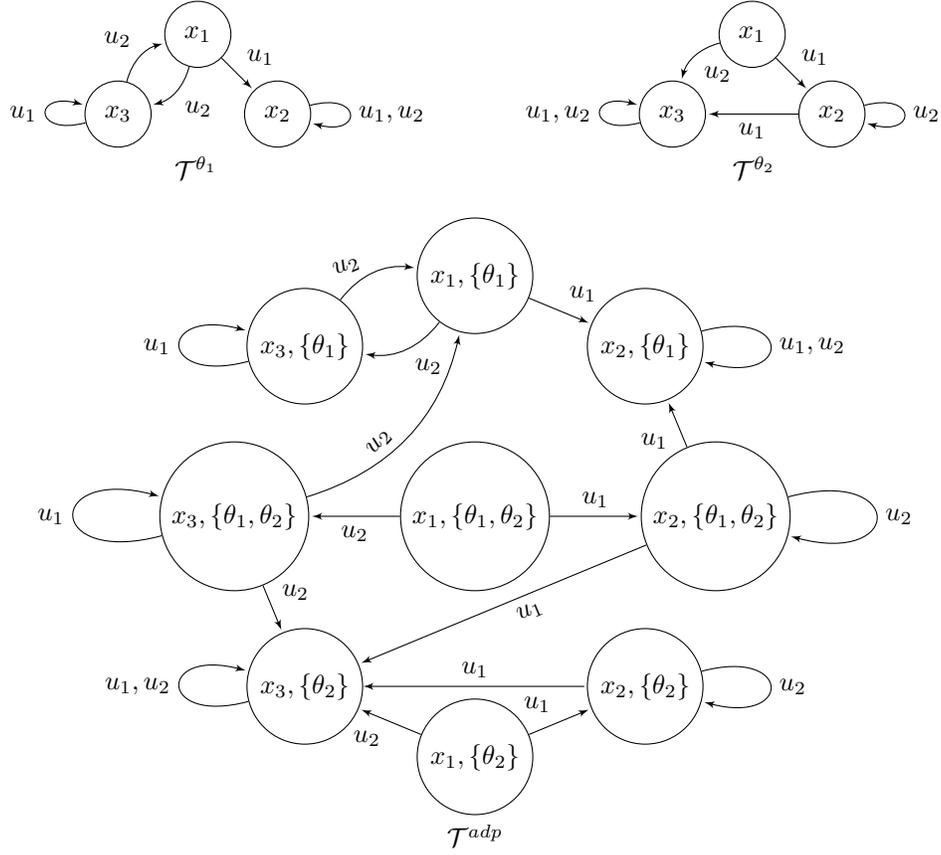

\subsection{Adaptive Transition System}
As mentioned in the introduction, a primary challenge of provably correct adaptive control is coupling parameter estimation and control synthesis. In order to combine these two, we provide the following definition.   

\begin{definition}
Given a PTS $\mathcal{T}^\Theta=\left(X, U, \Theta, \gamma, \Pi, O \right)$, we define the adaptive transition system (ATS) as the tuple $\mathcal{T}^{adp}=\left(X^{adp}, U, \gamma^{adp}, \Pi, O^{adp} \right)$, where $U,\Pi$ are inherited from $\mathcal{T}^{\Theta}$ with the same meaning and
\begin{itemize}
\item $X^{adp} \subseteq X \times 2^\Theta_{-\emptyset}$ is the set of states;
\item $\gamma^{adp}:X^{adp} \times U \rightarrow 2^{X^{adp}}$ is the transition function, where we have $(x',\vartheta') \in \gamma^{adp}((x,\vartheta),u)$ if and only if $x' \in \gamma(x,u)$ and $\vartheta' = \Gamma_{rec}(\vartheta,x,u,x')$;
\item $O^{adp}:X^{adp} \rightarrow 2^\Pi$ is the observation function where $O^{adp}(x,\vartheta)=O(x), \forall x\in X, \vartheta \in 2^\Theta_{-\emptyset}$.
\end{itemize} 

\begin{example}
Consider a PTS with $X=\{x_1,x_2,x_3\},U=\{u_1,u_2\},$ and $\Theta=\{\theta_1,\theta_2\}.$ The transition systems corresponding to $\theta_1$ and $\theta_2$ are illustrated in Fig. \ref{fig:ATS} [top]. The ATS corresponding is shown in Fig. \ref{fig:ATS} [Bottom].
\end{example}

\end{definition}
The number of states in the ATS is upper-bounded by $|X|(2^{|\Theta|}-1)$, which shows an exponential explosion with the number of parameters. Fortunately, not all states in $X \times 2^\Theta_{-\emptyset}$ are reachable from the set $\{(x,\theta) | x \in X,\theta \in \Theta \}$, which is the set of possible initial states in the ATS. Algorithm \ref{alg:ats} constructs the ATS consisting of only these reachable states.

\begin{algorithm}
\caption{Procedure for Constructing ATS from a PTS}
\label{alg:ats}
\begin{algorithmic}[0]
\Require{$\mathcal{T}^\Theta=\left(X, U, \Theta, \gamma, \Pi, O \right)$}

\State{$X^{adp,new}=\{ (x, \Theta) | x \in X \}$}

\State{$X^{adp}=X^{adp,new}$}

   \While{$X^{adp,new} \neq \emptyset$ }
%

      \State $ X^{adp,new} \gets \emptyset$

      \For {$ (x,\vartheta) \in  X^{adp}$}
      	\For {$ u \in  U$}
			\State $ \gamma^{adp}((x,\vartheta),u)=\emptyset$
			\State $ \vartheta'=\emptyset$
			\For {$\theta \in \vartheta$}
      			\For {$x' \in \gamma(x,u,\vartheta)$}
					\For {$\theta' \in \vartheta$}
     					\If  { $x' \in \gamma(x,u,\theta')$  }
							\State {$ \vartheta' \gets \vartheta' \cup \theta'$}
      					\EndIf
      
      				\EndFor
					\State{$\gamma^{adp}((x,\vartheta),u) \gets \gamma^{adp}((x,\vartheta),u) \cup (x',\vartheta')$}
					\If {$(x',\vartheta') \not \in X^{adp}$}
					\State{$X^{adp,new} \gets X^{adp,new} \cup (x',\vartheta')$}
					\State{$X^{adp} \gets X^{adp} \cup (x',\vartheta')$}
					\State{$O^{adp}(x',\vartheta')=O(x')$}
					\EndIf
      \EndFor
      \EndFor
      \EndFor
      \EndFor
   \EndWhile
\State \textbf{return} $\mathcal{T}^{adp}=\left(X^{adp}, U, \gamma^{adp}, \Pi, O^{adp} \right)$
\end{algorithmic}
\end{algorithm}

\subsection{Control Synthesis}

Finally, given an ATS $\mathcal{T}^{adp}$ and an LTL formula $\varphi$, we construct the product automaton $\mathcal{T}^{adp} \otimes \mathcal{R}_\varphi$ as explained in Sec. \ref{sec:ltlcon}, and find the memoryless control strategy on $\mathcal{T}^{adp} \otimes \mathcal{R}_\varphi$ by solving the Rabin game. We also find the largest set of admissible initial conditions $X_0^{adp,\max}$ as the winning region of the Rabin game. In order to find $X_0^{\max}$, we perform the following projection:
\begin{equation}
X_0^{\max}=\left\{x_0 \Big| (x_0,\Theta) \in X_0^{adp,\max} \right\}.
\end{equation}
The adaptive control strategy takes the memoryless form $\Lambda: X \times 2^\Theta_{-\emptyset} \times S \rightarrow U$, which maps the current state in the PTS, the set of current possible ground truth parameters and the state in the Rabin automaton to an admissible control action.

\begin{theorem}
Given a finite system \eqref{eq:system}, an initial condition $x_0 \in X$, an LTL formula over $\Pi$, there exists a control strategy $\Lambda^*: X^* \times U^* \rightarrow U$   such that $O(x_0)O(x_1)\cdots \models \varphi$, $\forall \theta \in \Theta, \forall d_k \in D$, $x_{k+1}=F(x_k,u_k,\theta,d_k), \forall k\in \mathbb{N}$, if and only if $x_0 \in X_0^{\max}$. . 
 \end{theorem}
\begin{proof}
(sketch) The completeness property follows from two facts. First, the solutions to Rabin games on finite automata are complete. Second, every possible behavior of a finite PTS embedding \eqref{eq:system} and parameter estimator \eqref{eq:estimator} is captured in the ATS. If $x_0 \not \in X_0^{\max}$, then it can be shown that there exists a $\theta \in \Theta$ and a disturbance sequence $d_0d_1\cdots$ such that there does not exist any control strategy to satisfy the LTL specification. 
\end{proof}

\section{Control Synthesis for Infinite Systems}
\label{sec:infinite}
In this section, we assume that PTS embedding \eqref{eq:system} is not finite, which means that at least one of the sets $X,U,\Theta$ is infinite. We provide the general solution for the case when all sets are infinite. We note that the approach in this section is still preliminary and we leave further investigation to our future work.

We consider a finite observation preserving (see Sec. \ref{sec:quotient}) partition $Q_X=\{q_X^1,\cdots,q_X^{p_X} \}$ for $X$ and a finite partition $Q_\Theta=\{q_\Theta^1,\cdots,q_\Theta^{p_\Theta} \}$ for $\Theta$. We also quantize $U$ to obtain a finite  $U_{\text{qtz}}=\{u_{qtz}^1,\cdots,u_{qtz}^{p_u} \}$. In this paper, we do not consider any particular guideline for how to partition and leave this problem to our future work. In general, the finer the partitions, the less conservative the method is with a price of higher computational effort. 
``Smart" partition refinement procedures were studied in \cite{yordanov2013formal,nilsson2014incremental}.  

Once partitions and quantizations are available, we compute the transitions. We denote the successor (post) of set $q_X$, under parameter set $q_\Theta$ and control $u$ by 
\begin{equation}
\label{eq:post}
\small
\text{Post}(q_X,q_\Theta,u) := \Big\{ x \in X \big | \exists x \in q_X, \exists \theta \in q_\Theta, x \in \gamma(x,\theta,u) \Big\}. 
\end{equation} 
A computational bottleneck is performing the post computation in \eqref{eq:post}. For additive parameters, the post computation is exact for piecewise affine systems using polyhedral operations \cite{Yordanov2012}. For multiplicative parameters, an over-approximations of post can be computed \cite{yordanov2008formal}, which introduces further conservativeness but retains correctness.    
Finally, we construct the quotient PTS from the infinite PTS. The procedure is outlined in Algorithm \ref{alg:quotient}.

\begin{algorithm}
\caption{Procedure for Constructing quotient PTS from infinite PTS}
\begin{algorithmic}[0]
\Require{$\mathcal{T}^\Theta=\left(X, U, \Theta, \gamma, \Pi, O \right)$}
\Require{$Q_X,Q_\Theta,U_{\text{quantized}}$}
\For{$q_X \in Q_X$}
	\State{$O^Q(q_X)=O(x)$ for some $x \in q_X$}
	\For{$q_\Theta \in Q_\Theta$}
		\For{$u_{qtz} \in U_{qtz}$}
			\State{$X_{\text{post}}=\text{Post}(q_X,q_\Theta,u)$}
			\State{$\gamma^{Q}(q_X,u_{qtz},q_\Theta)=\emptyset$}
			\For{$q_\Theta' \in Q_\Theta$}
			\If{$X_{\text{post}} \cap q'_\Theta \neq \emptyset$}
			\State{$\gamma^{Q}(q_X,u_{qtz},q_\Theta) \gets \gamma^{Q}(q_X,u_{qtz},q_\Theta) \cup q_\Theta'$}
			\EndIf
			\EndFor
		\EndFor
	\EndFor
\EndFor
\State \textbf{return} $\mathcal{T}^{Q,\Theta}=\left(Q, U_{\text{quantized}}, Q_\Theta, \gamma^{Q}, \Pi, O^Q \right)$
\end{algorithmic}
\label{alg:quotient}
\end{algorithm}


\section{Case Studies}
\label{sec:case}
We present two case studies. The first one is a simple finite deterministic system. The second case study involves a linear parameterized system that is infinite and non-deterministic due to the presence of additive disturbances.   

\subsection{Persistent Surveillance }
\begin{figure*}[t]
\begin{center}
\includegraphics[width=0.29\textwidth]{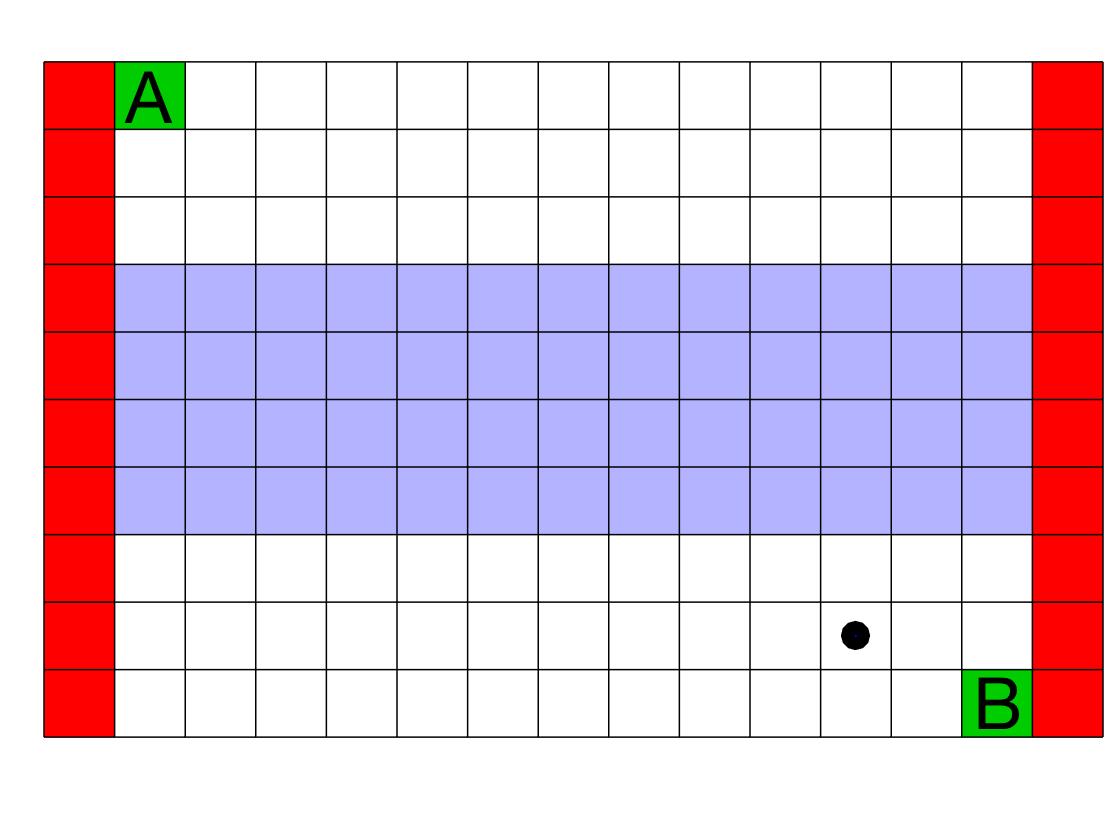}~
\includegraphics[width=0.29\textwidth]{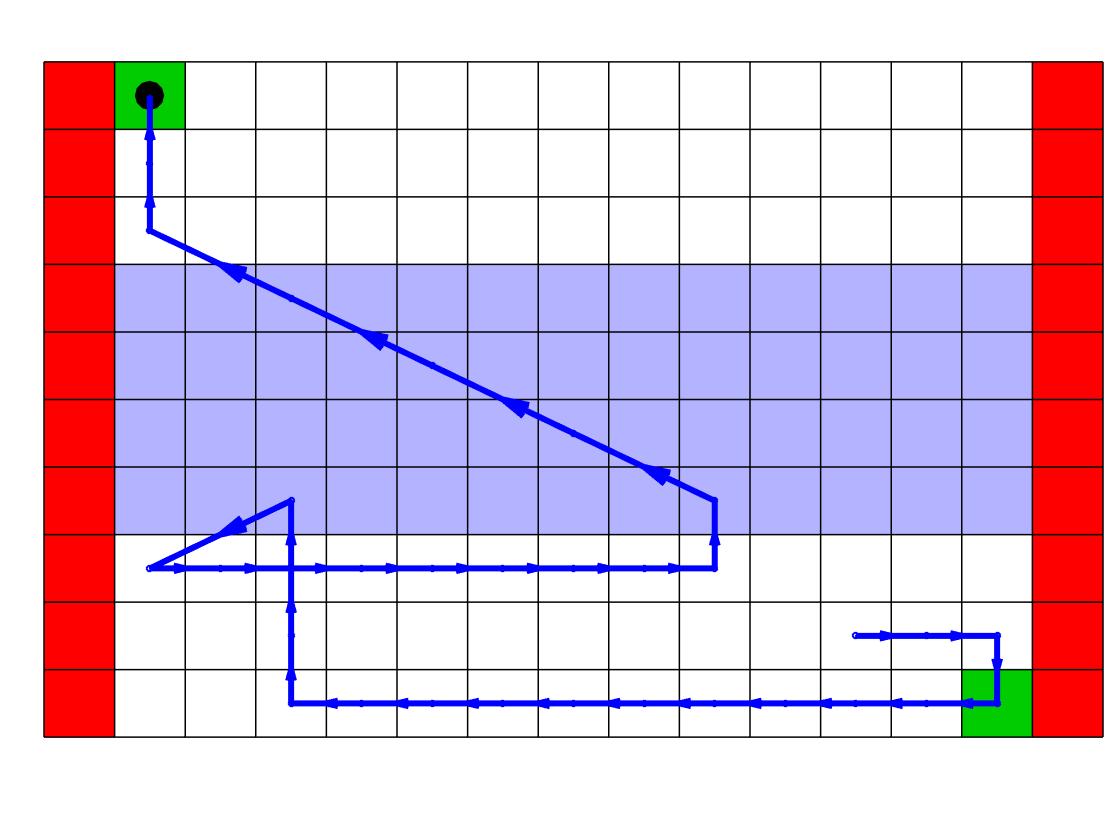}~\includegraphics[width=0.29\textwidth]{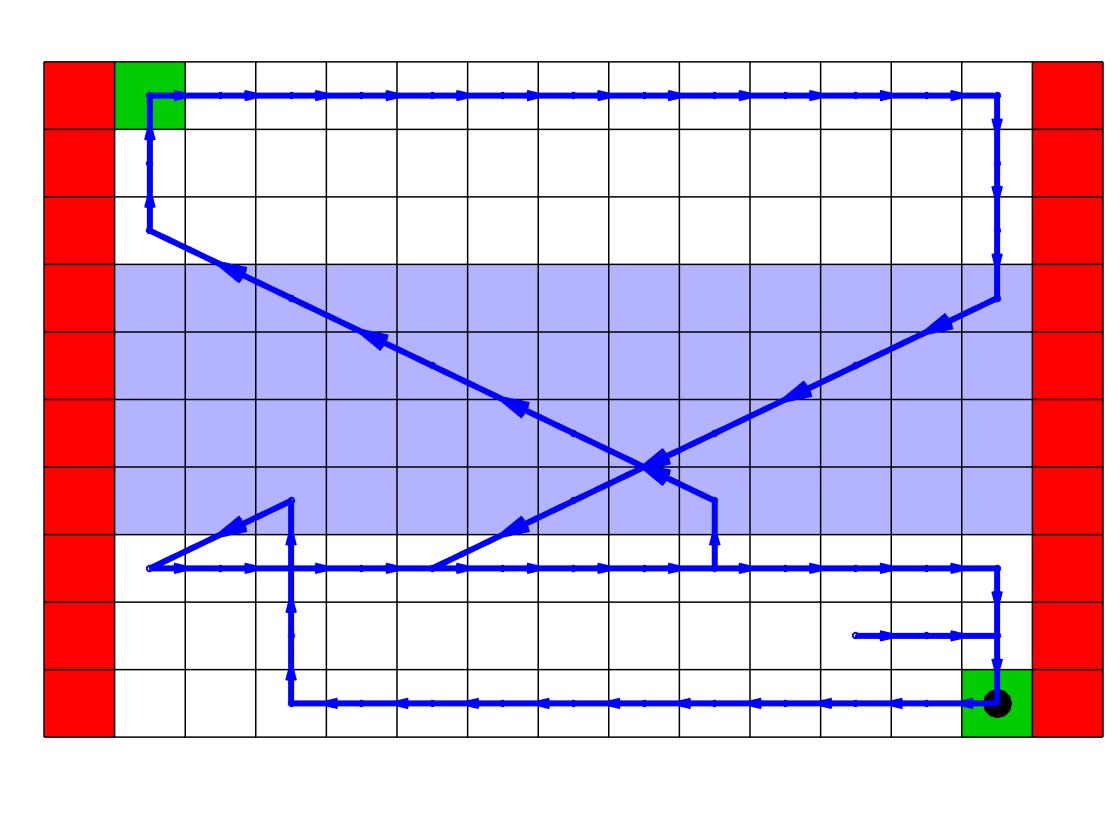}

\caption{Case Study 1: (Left): The Robot (shown in black) and its environment. 
(Middle): Snapshots of the executed Motion at time $k=33$, and (Right) $k=62$. The robot satisfies the specification.  
}
\label{fig:robot}
\end{center}
\end{figure*}

We consider a robot motion planning problem. The environment is modeled as a finite number of cells illustrated in Fig. \ref{fig:robot}. Each cell corresponds to a state in $X$. We have $|X|=150$. The set of control inputs is given by $U=\{$ {\bf left, right, up, down}$\}$, where the transition enabled by each input corresponds to its unambiguous meaning. There exists an constant drift in the horizontal direction in the purple region, but its direction to left or right and its intensity are unknown. The set of possible drifts is $\Theta=\{+2,+1,0,-1,-2\}$, where positive sign corresponds to the left direction. At each time, if the robot is in a purple cell, the drift is added to its subsequent position. For example, if the robot applies $u=${\bf right}, and $\theta^*=2$, the robot actually ends up in a cell to the left. Similarly, if $u=${\bf up} and $\theta^*=-2$, the robot moves a cell up and two cells to the right. The red cells are ``unsafe" regions that must be avoided, and the green cells $A,B$ are ``interesting" regions, which have to be persistently visited.  The LTL formula describing this specification is:  
\begin{equation*}
\varphi= {\bf G} {\bf F} A ~\wedge ~{\bf G} {\bf F} B ~\wedge~ {\bf G} (\neg \text{unsafe}).
\end{equation*}
We implemented the procedure outlined in Sec. \ref{sec:finite}. It is worth to note that there does not exist a pure robust control solution to this problem. In other words, if the robot ignores estimating the drift, it can not find a control strategy. For example, if the robot enters the purple region around the middle and persistently applies ${\bf up}$, a maximum drift in either direction can drive the robot into the unsafe cells before it exits the purple region. Therefore, the only way the robot can fulfill the specification is to learn the drift. The robot first enters the drifty region to find out its value and then moves back and re-plans its motion. Notice that this procedure is fully automated using the solution of the Rabin game on the product $\mathcal{T}^{adp} \otimes \mathcal{R}_\varphi$. Two snapshots of the executed motion for the case $\theta^*=+2$ are shown in Fig. \ref{fig:robot}.

\subsection{Safety Control}
\label{sec:safety}
Consider a one-dimensional linear system of the following form:
\begin{equation}
x^+=(1+\theta_1)x+\theta_2u+\theta_3+d,
\end{equation}
where $\theta_1 \in [-0.5,0.5]$, $\theta_2 \in [1,2]$, and $\theta_3 \in [-0.2,0.2]$ are fixed parameters, and $d\in D$, is the additive disturbance, $D= [-0.1,0.1]$. The set of admissible control inputs is $U=[-1,1]$. We desire to restrict $x$ to the $[-1,1]$ interval for all times, which is described by the following LTL formula:  
\begin{equation*}
\varphi= {\bf G} (x \le 1) \wedge {\bf G} (x \ge -1).
\end{equation*}
We have $\Theta=[-0.5,0.5] \times [1,2] \times [-0.2,0.2]$. 
We partitioned the intervals of $\theta_1$, $\theta_2$, $\theta_3$, and $X$ into 2,2,4, and 10 evenly spaced intervals, respectively. Thus, we have partitioned $\Theta$ into $16$ cubes ($|Q_\Theta|=16$) and $X$ into 10 intervals ($|Q_X|=10$). $U$ is quantized to obtain $U_{qtz}=\{-1,-0.8,\cdots,0.8,1\}$. We implemented Algorithm \ref{alg:quotient} to obtain the quotient PTS and Algorithm \ref{alg:ats} to find the corresponding ATS. The computation times were 0.1 (Algorithm \ref{alg:quotient}) and 152 (Algorithm \ref{alg:ats}) seconds on a 3.0 GHz MacBook Pro.   
Even though $|X \times 2_{-\emptyset}^{Q_\Theta}|=655350$, the number of reachable states obtained from Algorithm \ref{alg:ats} was 14146.    

We solved the safety game on the ATS, which took less than a second and found a winning region containing 14008 states. The winning region in the state-space is $X_0=[-0.6,0.6]$. Since the solution is conservative, $X_0^{\max}$ may be larger if a finer partitioning is used. We also found that the winning region is empty if we had sought a pure robust control strategy. 
 We simulated the system for 100 time steps starting from $x_0=0$. The values of disturbances at each time are chosen randomly with a uniform distribution over $D$. We observe that the specification is satisfied, and the sets given by the parameter estimator shrink over time and always contain the ground truth parameter, which in this case is $\theta_1^*=0.45$, $\theta_2^*=1.11$, $\theta_3^*=-0.18$. The results are shown in Fig. \ref{fig:safety}.

\begin{figure}[t]
\centering
\includegraphics[height=0.18\textwidth]{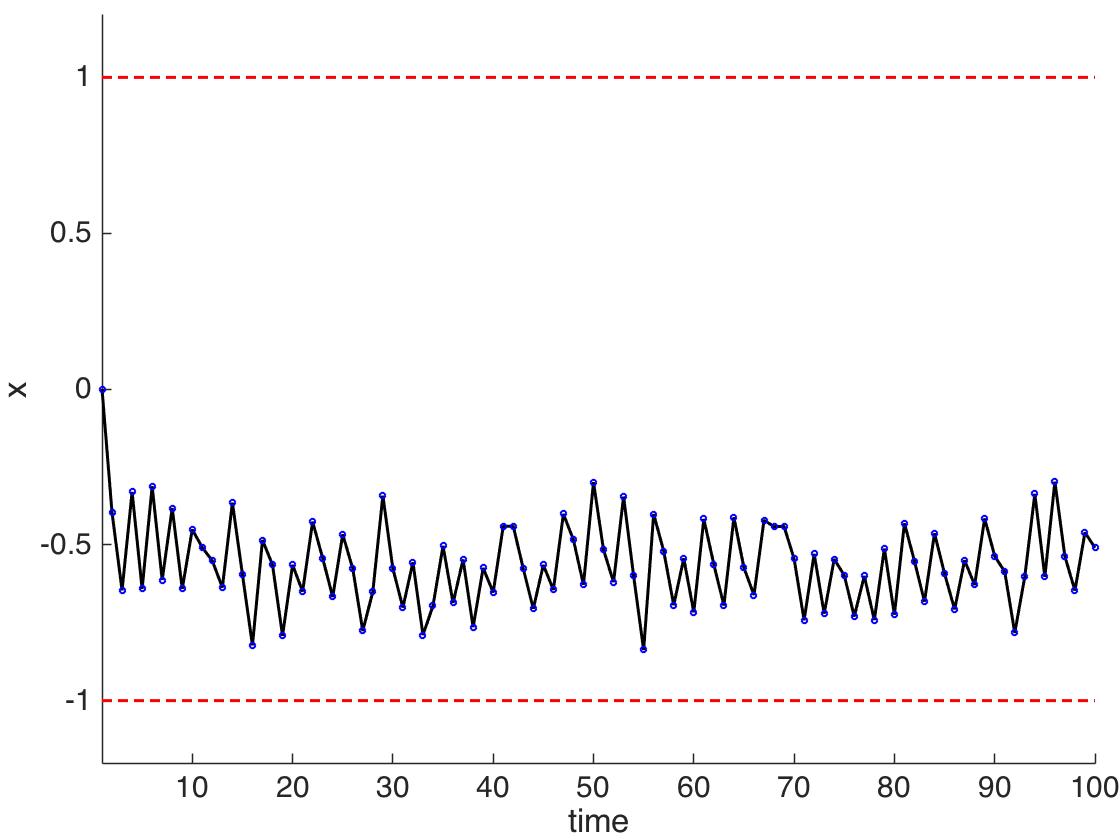}
\caption{Case Study 2: trajectory of the system versus time, which is always between $-1$ and $1$. }
\end{figure}

\begin{figure}[t]
\centering
\vspace{0.2in}
\includegraphics[width=0.24\textwidth]{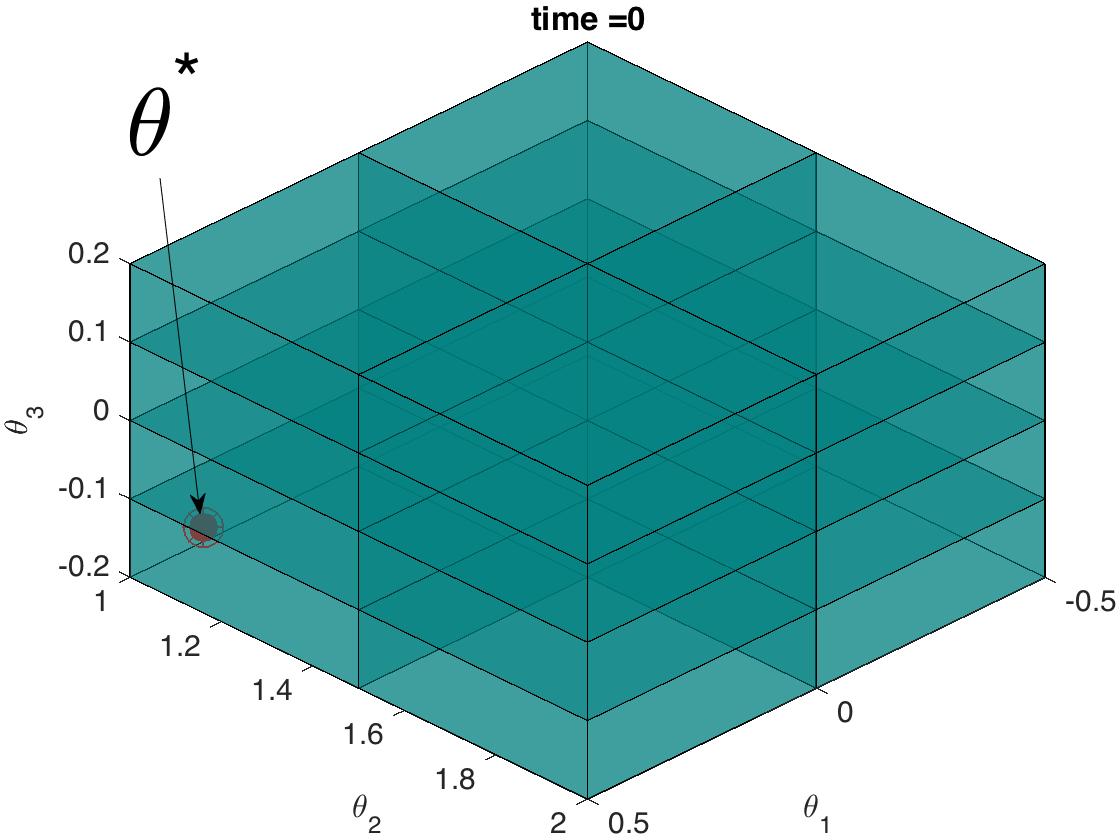}\includegraphics[width=0.24\textwidth]{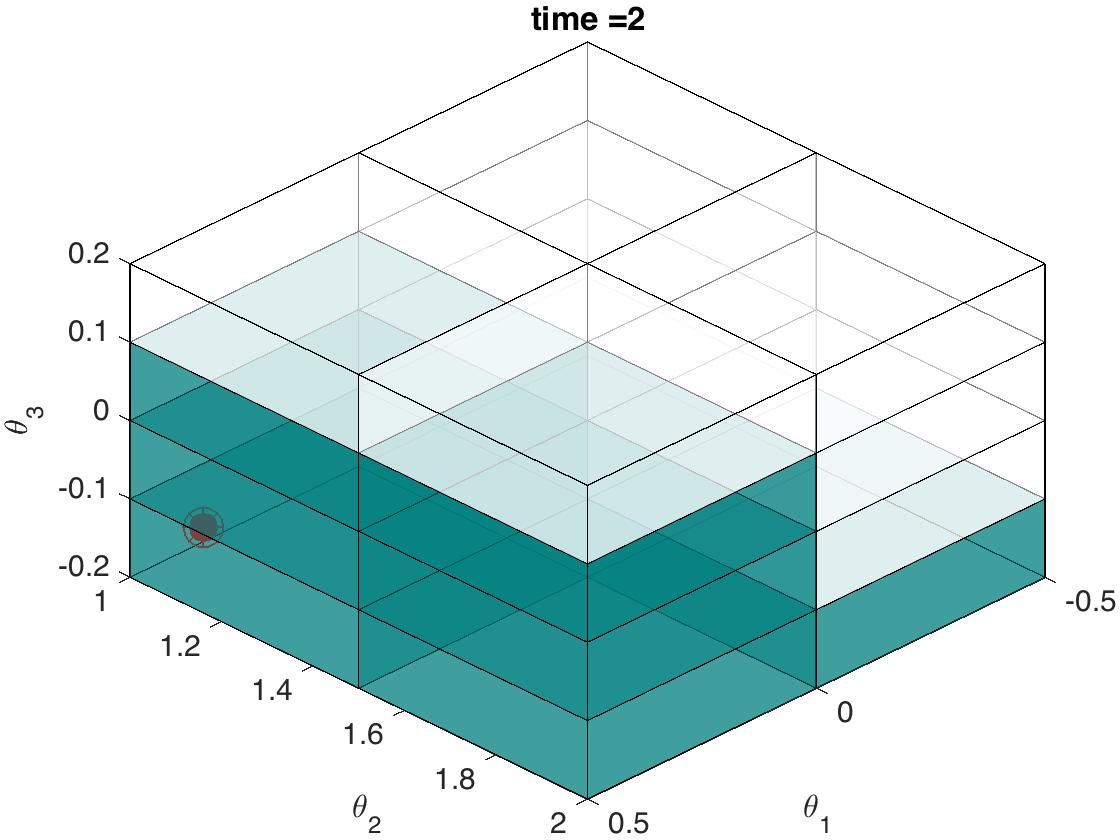}
\\~\includegraphics[width=0.24\textwidth]{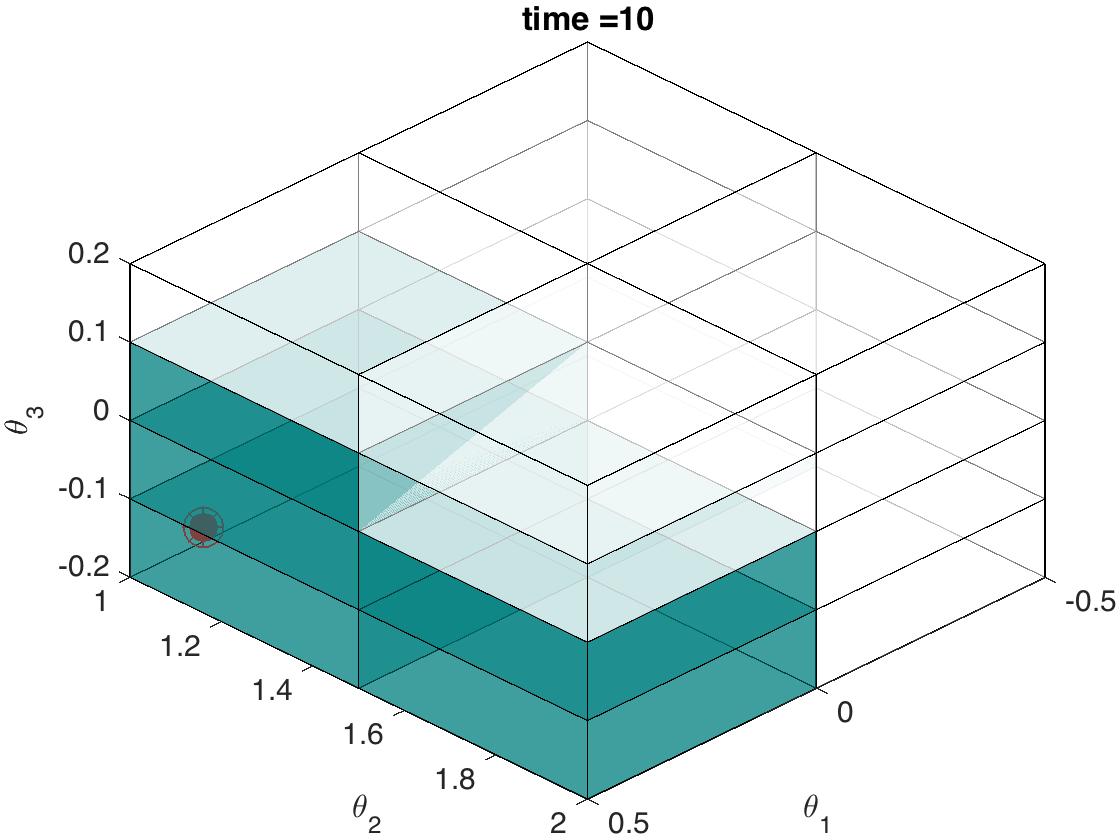}\includegraphics[width=0.24\textwidth]{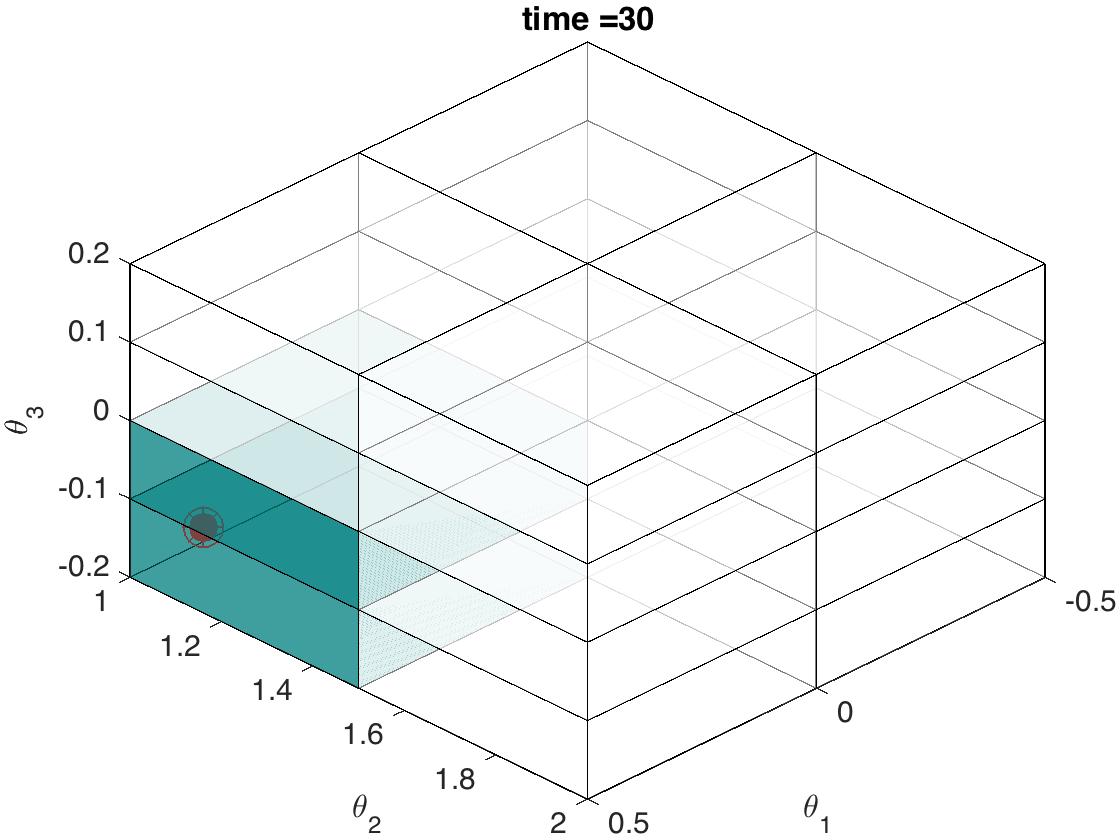}~
\caption{Case Study 2: Snapshots of  $\vartheta_k$ at various times, which are illustrated by the shaded regions. They always contain the ground truth parameter $\theta_1^*=0.45$, $\theta_2^*=1.11$, $\theta_3^*=-0.18$.}
\label{fig:safety}

\end{figure}

\section{Conclusion and Future Work}
We developed a framework to combine the recent advances in applications of formal methods in control theory with classical adaptive control. We used the concepts from transition systems, finite quotients, and product automata to introduce adaptive transition systems and correct-by-design adaptive control. Like most of other formal methods applications, our results suffer from high computational complexity. As discussed in the paper, the number of states in the ATS can be very large. Also, constructing finite quotients for infinite systems is computationally difficult.  
    
We believe that this paper opens up several research directions. Besides improving the ideas for the way we combine adaptive control and formal methods, we plan to develop efficient methods to construct finite adaptive transition systems for special classes of hybrid systems such as mixed-monotone systems and piecewise affine systems. We also plan to include optimal control.

\balance

\bibliographystyle{IEEEtran}
\bibliography{references2}  

\end{document}